\documentclass{lmcs} 
\pdfoutput=1
\usepackage[utf8]{inputenc}

\usepackage{lastpage}
\lmcsdoi{20}{1}{1}
\lmcsheading{}{\pageref{LastPage}}{}{}%
{Oct.~11,~2022}{Jan.~11,~2024}{}


\usepackage{hyperref}
\usepackage{graphicx}
\usepackage{mathrsfs}
\usepackage{color}
\usepackage{amsmath}
\usepackage{amssymb}
\usepackage{amsfonts}
\usepackage{algorithmic}
\usepackage{nicefrac}
\usepackage{cleveref}

\Crefname{lem}{Lemma}{Lemma}
\Crefname{thm}{Theorem}{Theorem}
\Crefname{fact}{Fact}{Fact}
\Crefname{thmC}{Theorem}{Theorem}

\newtheorem{definition}[thm]{Definition}   
\DeclareMathOperator{\SAT}{\mathsf{SAT}}
\DeclareMathOperator{\CNF}{\mathsf{CNF}}
\DeclareMathOperator{\DPLL}{\mathsf{DPLL}}
\DeclareMathOperator{\ILP}{\mathsf{ILP}}
\DeclareMathOperator{\LP}{\mathsf{LP}}

\DeclareMathOperator{\SP}{\mathsf{SP}}
\DeclareMathOperator{\Res}{\mathsf{Res}}
\DeclareMathOperator{\CP}{\mathsf{CP}}
\DeclareMathOperator{\UNSAT}{\mathsf{UNSAT}}
\DeclareMathOperator{\slab}{\mathsf{slab}}
\DeclareMathOperator{\PHP}{\mathsf{PHP}}
\DeclareMathOperator{\LOP}{\mathsf{LOP}}
\DeclareMathOperator{\SPHP}{\mathsf{SPHP}}
\DeclareMathOperator{\TS}{\mathsf{Ts}}
\DeclareMathOperator{\VER}{\mathsf{VER}}
\DeclareMathOperator{\IND}{\mathsf{IND}}
\DeclareMathOperator{\Peb}{\mathsf{Peb}}

\newcommand{\T}{{\mathcal T}}
\newcommand{\F}{{\mathcal F}}
\newcommand{\sF}{{\mathscr F}}
\newcommand{\A}{{\mathcal A}}
\newcommand{\R}{{\mathcal R}}

\begin{document}

\title[Depth lower bounds in Stabbing Planes for combinatorial principles]{Depth lower bounds in Stabbing Planes\texorpdfstring{\\}{ }for combinatorial principles\rsuper*}
\titlecomment{{\lsuper*}An extended abstract of this paper has appeared at The 39th International Symposium on Theoretical Aspects of Computer Science (STACS 2022)  \cite{DantchevGGM22}}

\author[S.~Dantchev]{Stefan Dantchev\lmcsorcid{0000-0002-4534-2242}}[a]
\author[N.~Galesi]{Nicola Galesi\lmcsorcid{0000-0002-8522-362X}}[b]
\author[A.~Ghani]{Abdul Ghani}[a]
\author[B.~Martin]{Barnaby Martin\lmcsorcid{0000-0002-4642-8614}}[a]

\address{Durham University, Durham, DH1 3LE. UK}	
\email{s.s.dantchev@durham.ac.uk, barnaby.d.martin@durham.ac.uk}  

\address{Sapienza Universit\`a  di Roma, Rome, Italy}	
\email{galesi@di.uniroma1.it}  





\begin{abstract}
{\em Stabbing Planes} (also known as Branch and Cut) is a proof system introduced very recently which, informally speaking, extends the $\DPLL$ method by branching on integer linear inequalities instead of single variables.
The techniques known so far to prove size 
and depth lower bounds for Stabbing Planes are generalizations of those used for the {\em Cutting Planes} proof system. For size lower bounds these are established by monotone circuit arguments, while for depth these are found via communication complexity and protection. As such these bounds apply for lifted versions of combinatorial statements. Rank lower bounds for Cutting Planes are also obtained by geometric arguments called {\em protection lemmas}.
  
In this work we  introduce two  new geometric approaches to prove size/depth lower bounds in Stabbing Planes working for any formula: (1) the {\em antichain method}, relying on {\em Sperner's Theorem} and (2) the {\em covering method} 
which uses results on {\em essential  coverings} of the boolean cube by linear polynomials, which in turn  relies on {\em Alon's combinatorial Nullenstellensatz}.  

We demonstrate their use on  classes of  combinatorial principles such as the {\em Pigeonhole principle}, the {\em Tseitin} contradictions and the {\em Linear Ordering Principle}. 
By the first method we prove almost linear size lower bounds and optimal logarithmic depth lower bounds for the Pigeonhole principle and analogous {lower bounds} 
for the Tseitin contradictions over the complete graph and for the Linear Ordering Principle. 
By the  covering method we obtain  a superlinear  size lower bound and a logarithmic depth lower bound for Stabbing Planes proof of Tseitin contradictions over a grid graph.
\end{abstract}

\maketitle

\section{Introduction}
\label{intro}

Finding a satisfying assignment for a propositional formula ($\SAT$) is a central component for many computationally hard problems. Despite being older than 50 years and exponential time in the worst-case, the  $\DPLL$  algorithm \cite{DBLP:journals/cacm/DavisLL62,DBLP:journals/jacm/DavisP60,DBLP:journals/jacm/Robinson65} is the core of essentially all high performance modern $\SAT$-solvers. $\DPLL$ is a recursive boolean method: at each call one variable $x$ of the formula $\mathcal{F}$  is chosen and the search recursively branches into the two cases  obtained by setting $x$ respectively to $1$ and $0$ in $\mathcal{F}$. On $\UNSAT$ formulas  $\DPLL$ performs the worst and it is well-known that the execution trace of the $\DPLL$ algorithm  running on an unsatisfiable formula $\mathcal{F}$ is nothing more than a treelike refutation of $\mathcal{F}$ in the proof system of {\em Resolution} \cite{DBLP:journals/jacm/Robinson65} ($\Res$). 

Since $\SAT$ can be viewed as an optimization problem the question whether Integer Linear Programming ($\ILP$) can  be  made feasible for satisfiability testing received a lot of attention and is considered among the most challenging problems in local search \cite{DBLP:conf/ijcai/SelmanKM97,DBLP:conf/cp/KautzS03}.   One proof system capturing $\ILP$ approaches to $\SAT$ is {\em Cutting Planes}, a system whose main  rule implements the {\em rounding} (or {\em Chv\'atal cut}) approach to $\ILP$. Cutting planes works with  integer linear inequalities of the form $\mathbf a \mathbf x \leq b$, with $\mathbf a,b$ integers, and, like resolution, is a sound and complete refutational proof system for $\CNF$ formulas: indeed a clause $C=(x_1\vee\ldots \vee x_r\vee \neg y_1 \vee \ldots \vee \neg y_s)$ can be written as the integer inequality  $\mathbf y - \mathbf x  \leq  s-1$.     

Beame et al. \cite{DBLP:conf/innovations/BeameFIKPPR18}, extended the idea of $\DPLL$  to a more general proof strategy based on $\ILP$.  Instead of branching only on a variable as in resolution, in this method one considers a pair $(\mathbf a ,b)$, with $\mathbf a \in \mathbb Z^n$ and $b \in \mathbb Z$,  and branches limiting the search to the two half-planes: $\mathbf a \mathbf x \leq  b-1$ and $\mathbf a \mathbf x \geq  b$.
A {\em path} terminates when the $\LP$ defined by the inequalities in $\mathcal{F}$ and those forming the path is infeasible. 
This method  can be  made into a refutational treelike proof system for $\UNSAT $ CNF's called {\em Stabbing planes} ($\SP$) (\cite{DBLP:conf/innovations/BeameFIKPPR18})   and it turned out that it is  polynomially equivalent to the treelike version of $\Res(\CP)$, a proof system introduced by Kraj\'i\v{c}ek \cite{DBLP:journals/jsyml/Krajicek98}  where clauses are disjunction of linear inequalities. \textcolor{black}{Furthermore, Stabbing Planes captures the popular branch-and-cut $\ILP$ algorithms.} 

In this work we consider the complexity of proofs in $\SP$ focusing on the {\em length}, i.e. the number of queries in the proof;  the {\em depth} (called also {\em rank} in \cite{DBLP:conf/innovations/BeameFIKPPR18}), i.e. the length of the longest path in the proof tree; and the {\em size},  i.e. the bit size of all the coefficients appearing in the proof.

\subsection{Previous works and motivations}
After its introduction as a proof system in the work \cite{DBLP:conf/innovations/BeameFIKPPR18} by Beame, Fleming, Impagliazzo,  Kolokolova,  Pankratov,  Pitassi and Robere, {\em Stabbing Planes} received great attention. The quasipolynomial upper bound for the size of refuting Tseitin contradictions in $\SP$ given in \cite{DBLP:conf/innovations/BeameFIKPPR18} was surprisingly extended to $\CP$  in the work of  \cite{DBLP:conf/coco/DadushT20} of Dadush and Tiwari refuting a long-standing conjecture.  Recently in \cite{Fleming21}, Fleming, G\"{o}\"{o}s, Impagliazzo, Pitassi, Robere, Tan and Wigderson were further developing the initial results proved in \cite{DBLP:conf/innovations/BeameFIKPPR18}
 making important progress on the question whether all Stabbing Planes proofs can be somehow efficiently simulated by Cutting Planes.
 
Significant lower bounds for \textcolor{black}{depth} can be obtained \textcolor{black}{for} $\SP$, using modern developments of  a technique for $\CP$ based on communication complexity of search problems  introduced by Impagliazzo, Pitassi, Urquhart in \cite{impagliazzo1994upper}: 
in \cite{DBLP:conf/innovations/BeameFIKPPR18} it is proven that  size $S$ and \textcolor{black}{depth} $D$ $\SP$ refutations imply treelike $ \Res(\CP)$ proofs of size $O(S)$ and width $O(D)$; Kojevnikov \cite{DBLP:conf/sat/Kojevnikov07}, improving the {\em interpolation method} introduced for $\Res(\CP)$ by Kraj\'i\v{c}ek \cite{DBLP:journals/jsyml/Krajicek98}, gave   exponential lower bounds for treelike $\Res(\CP)$ when the width of the clauses (i.e. the number of linear inequalities in a clause) is bounded by $o(n/ \log n)$. \textcolor{black}{However \cite{DBLP:conf/innovations/BeameFIKPPR18} shows that there are no $n/\log n$ depth treelike $ \Res(\CP)$ proofs of the given formula at all}.
Hence these lower bounds are applicable only to very specific classes of formulas (whose hardness comes from boolean circuit hardness) and  only to $\SP$ refutations of low depth.   

Nevertheless $\SP$ appears to be a strong proof system. Firstly notice that the condition terminating a path in a proof
is not a trivial contradiction like in resolution, but is the infeasibility of an $\LP$, which is only a polynomial time verifiable condition. Hence linear size $\SP$ proofs might be already a strong class of $\SP$ proofs, since they can hide a polynomial growth into one final node whence to run the verification of the terminating condition. 

\subsubsection*{Rank and depth in $\CP$ and $\SP$}

It is  known that, contrary to the case of other proof systems like Frege, neither $\CP$ nor $\SP$ proofs can be balanced (see \cite{DBLP:conf/innovations/BeameFIKPPR18}), in the sense that a depth-$d$ proof can always be transformed into a size $2^{O(d)}$ proof. The depth of $\CP$-proofs of a set of linear inequalities ${L}$ is measured by the {\em Chv\'atal rank} of the associated polytope $P$.\footnote{This is the minimal $d$ such that  $P^{(d)}$ is empty, where $P^{(0)}$ is the polytope associated to $L$ and $P^{(i+1)}$ is the polytope defined by all inequalities which can be inferred from those in $P^{(i)}$ using one Chv\'atal cut.} It is known that rank in $\CP$ and depth in $\SP$ are separated, in the sense that Tseitin principles can be proved in  depth $O(\log^2 n)$ depth in  $\SP$ \cite{DBLP:conf/innovations/BeameFIKPPR18}, but are known to require rank $\Theta(n)$ to be refuted  in $\CP$ \cite{5authors}. In this paper we further develop the study of proof depth for $\SP$.

Rank lower bound techniques for Cutting Planes are essentially of two types. The main method is by reducing to the real communication complexity of certain search problem \cite{impagliazzo1994upper}. As such this method
only works for classes of formulas {\em lifted} by certain gadgets capturing specific  boolean functions.  
A second class of methods have been developed for Cutting Planes, which lower bound
the rank measures of a polytope. In this setting, lower bounds are typically
proven using a geometric method called {\em protection lemmas} \cite{5authors}. These methods were recently extended in \cite{Fleming21} also to the case of Semantic Cutting Planes.  In principle this  geometric  method  can be applied to any  formula and not only to the lifted ones, furthermore for many formulas (such as the Tseitin formulas)
it is known how to achieve $\Omega(n)$ rank  lower bounds in $\CP$ via protection lemmas, while
proving even $\omega(\log n)$ lower bounds via real communication complexity is impossible, due to a known
folklore upper bound.

Lower bounds for depth in Stabbing Planes, proved in \cite{DBLP:conf/innovations/BeameFIKPPR18},  are instead obtained only as a consequence of the real communication approach extended to Stabbing Planes.   In this paper we introduce  two  geometric approaches to prove depth lower bounds in $\SP$. 

Specifically the results we know at present relating $\SP$ and $\CP$ are:
 
\begin{enumerate}
\item $\SP$ polynomially simulates $\CP$ (Theorem 4.5 in \cite{DBLP:conf/innovations/BeameFIKPPR18}). Hence in particular the $\PHP^m_n$ can be refuted in $\SP$  by a proof of size  $O(n^2)$ (\cite{CookCoullardTuran}). Furthermore it can be refuted  by a 
$O(\log n)$ depth proof since polynomial size $\CP$ proofs, by Theorem 4.4 in \cite{DBLP:conf/innovations/BeameFIKPPR18}, can be balanced in $\SP$.\footnote{Another way of proving this result is using  Theorem  4.8 in \cite{DBLP:conf/innovations/BeameFIKPPR18}  stating that if there are length $L$ and  space $S$ $\CP$ refutations of a set of linear integral inequalities, then there are  depth $O(S \log L)$ $\SP$ refutations of the same set of linear integral inequalities;  and then use the result  in \cite{DBLP:conf/coco/GalesiPT15}  (Theorem 5.1) that $\PHP^m_n$ has polynomial length and constant space $\CP$ refutations.}

\item Beame et al. in \cite{DBLP:conf/innovations/BeameFIKPPR18} proved the surprising result that the class of Tseitin contradictions $\TS(G,\omega)$ over any graph $G$ of maximum degree $D$, with an odd charging $\omega$, can be refuted in $\SP$ in size  quasipolynomial in $|G|$ and depth $O(\log^2|G|+D)$.  

\item \textcolor{black}{Fleming et al. in \cite{Fleming21} proved that a size $S$ (and maximal coefficient size $C$) $\SP$
refutation of a unsatisfiable formula $F$ over $n$ variables can be converted  into a $\CP$ refutation of $F$ of size
$S(Cn)^{\log S}$. However in this case the depth of the proof may potentially blow-up as well.} 

\end{enumerate}
Depth lower bounds for $\SP$ are proved  in  \cite {DBLP:conf/innovations/BeameFIKPPR18}:  
\begin{enumerate}
\item a $\Omega(n/\log^2n)$ lower bound for the formula 
$\TS(G,w) \circ \VER^n$,  composing $\TS(G,\omega)$ (over an expander graph $G$)  with the gadget function $\VER^n$ 
(see Theorem 5.7 in \cite{DBLP:conf/innovations/BeameFIKPPR18} for details); and 
\item a $\Omega(\sqrt{n	\log n})$  lower bound for the formula $\Peb(G) \circ \IND^n_l$ over $n^5+n\log n$ variables obtained by lifting a pebbling formula $\Peb(G)$ over a graph with high pebbling number, with a {\em pointer function} gadget $\IND^n_l$ (see Theorem   5.5. in \cite{DBLP:conf/innovations/BeameFIKPPR18} for details).
\end{enumerate}

Similar to size, these depth lower bounds are  applicable only to very specific classes of formulas.  In fact they are obtained by extending to $\SP$ the technique introduced in  \cite{impagliazzo1994upper,DBLP:journals/mlq/Krajicek98} for $\CP$ of  
reducing shallow proofs of a  formula $\mathcal{F}$ to efficient {\em real} communication protocols  computing a related search problem and then proving that such efficient protocols cannot exist. 

\textcolor{black}{The only lower bounds techniques on the depth of Stabbing Planes proofs come from reductions to communication complexity, which is a lower bound technique for $\CP$. This is also in contrast with other weaker proof systems such as Resolution and Cutting Planes, where we have direct combinatorial and geometric techniques for proving depth lower bounds. Direct lower bound techniques are valuable as they are tailored to the proof system and thus shed light on its behaviour and weaknesses, unlike semantic techniques such as reductions to monotone circuits or communication complexity, which prove lower bounds on more general objects (such as monotone circuits and real communication protocols).}


In this work we address such problems.

\subsection{Contributions and techniques}
The main motivation of this work was to study size and depth lower bounds in $\SP$ through new methods, possibly geometric. Differently from weaker systems like Resolution, except for the technique highlighted above and based on reducing  to the communication complexity of search problems, 
we do not know of other methods to prove size and depth lower bounds in  $\SP$. In $\CP$  and  Semantic $\CP$ instead geometrical methods based on protection lemmas were used to prove rank lower bounds in \cite{5authors,Fleming21}.

Our first steps in this direction were to set up methods working for truly combinatorial statements, like $\TS(G,w)$ or $\PHP^m_n$, which we know to be efficiently provable in $\SP$, but on which we cannot  use methods reducing to the complexity of boolean functions, like the ones based on communication complexity. 

We present two new methods for proving depth lower bounds in $\SP$ which in fact are the consequence of proving length lower bounds that do not depend on the bit-size of the coefficients. 

As applications of our two methods we respectively prove:  
\begin{enumerate}
 
\item An exponential separation between the rank\footnote{\textcolor{black}{The distinction between $\CP$ depth and rank comes from the latter only counting applications of the rounding rule.}} in $\CP$ and the depth in $\SP$, using a new  counting principle which we introduce and that we call the {\em Simple Pigeon Principle} $\SPHP$. We prove that $\SPHP$ has $O(1)$ rank in $\CP$ and requires $\Omega(\log n)$ depth in $\SP$. Together with the results proving that Tseitin formulas requires $\Omega(n)$ rank lower bounds in $\CP$ (\cite{5authors})  and 
$O(\log^2n)$ upper bounds for the depth in $\SP$ (\cite{DBLP:conf/innovations/BeameFIKPPR18}), this proves an incomparability between the two measures.  

\item An almost linear lower bound on the size of $\SP$ proofs of the $\PHP^m_n$ and  for Tseitin  $\TS(G,\omega)$ contradictions over the complete graph. These lower bounds immediately give optimal $\Omega(\log n)$ lower bound for the depth of $\SP$ proofs of the corresponding principles. 
\item An almost linear lower bound for the size and \textcolor{black}{$\Omega(\log n)$} lower bound of the depth for the the Linear Ordering Principle $\LOP_n$.  

\item Finally, we prove a superlinear lower bound for the size of $\SP$ proofs   of  $\TS(G,\omega)$, when $G$ is a $n \times n$ grid graph  $H_n$. In turn this implies an $\Omega(\log n)$ lower bound for the depth of $\SP$ proofs of  $\TS(H_n,\omega)$. Proofs of depth $O(\log^2 n)$  for $\TS(H_n,\omega)$ are given in \cite{DBLP:conf/innovations/BeameFIKPPR18}.

\end{enumerate}

Our results are derived from the following initial geometrical observation: let $\mathbb S$ be a  space of {\em admissible points}
in $\{0,1,1/2\}^n$ satisfying a given unsatisfiable system of integer linear inequalities $\mathcal{F}(x_1,\ldots,x_n)$. In a $\SP$ proof for $\mathcal{F}$, at each 
branch $Q=(\mathbf a, b)$ the set of points in the $\slab(Q)=\{\mathbf s \in \mathbb S :  b-1< \mathbf a \mathbf x < b\}$ does not survive.  At the end of the proof on the leaves, where we have infeasible $\LP$'s, no point in $\mathbb S$ can  survive the proof.  So it is  sufficient  to find conditions such that, under the assumption that  a proof of $\mathcal{F}$  is ``small'',  even one point of $\mathbb S$  survives the proof. In pursuing this approach we use two methods.  


The {\em antichain method}. Here we use a well-known bound based on Sperner's Theorem \cite{anti2,lint01} to upper bound the number of points in the slabs where the set of non-zero coefficients is sufficiently large.  Trading between the number of such slabs and the number of points ruled out from the space  $\mathbb S$ of admissible points, we obtain the  lower bound.  

We initially present the method and the $\Omega(\log n)$  lower bound on a set of unsatisfiable integer linear inequalities - the {\em Simple Pigeonhole Principle} ($\SPHP$) - capturing the core of the counting argument used to prove the PHP efficiently in CP.  Since $\SPHP_n$ has rank $1$ $\CP$ proofs, it entails  a strong separation between $\CP$ rank and $\SP$ depth. 
We then apply  the method to $\PHP^m_n$ and to $\TS(K_n,\omega)$.

The {\em covering method}. The antichain method appears too weak to prove size and depth lower bounds  on $\TS(G,w)$, when $G$ is for example a grid or a pyramid.  To solve this case,  we consider another approach that we call the {\em covering method}: we reduce the problem of proving that one point in $\mathbb S$ survives from all the $\slab(Q)$ in a small proof of  $\mathcal{F}$, to the problem that a set of  polynomials which {\em essentially covers} the boolean cube  $\{0,1\}^n$  requires at least $\sqrt{n}$ polynomials, which is a well-known problem faced by Alon and  F\"uredi  in \cite{DBLP:journals/ejc/AlonF93} and by  Linial and Radhakrishnan in \cite{essential}. 
For this reduction to work we have to find a high dimensional projection of $\mathbb S$ covering the boolean cube and defined on variables effectively appearing in the proof.  We prove that cycles of distance at least 2 in $G$ work properly to this aim on $\TS(G,\omega)$. Since the grid $H_n$ has many such cycles, we can obtain
the lower bound on  $\TS(H_n,\omega)$. The use of Linial and Radhakrishnan's result is not new in proof complexity.
Part and Tzameret in \cite{DBLP:journals/cc/PartT21}, independently of us, \textcolor{black}{were using this result in a similar way to us to prove size lower bounds in the proof system  $\Res(\oplus)$ over integers which handles clauses over  linear  equations, and not relying on  integer linear inequalities and geometrical reasoning}. 

\smallskip
We  remark that while we were writing this version of the  paper, Yehuda and Yehudayoff in  \cite{YY21}  slightly improved 
the results of \cite{essential} with the consequence, noticed in their paper too, that our size lower bounds for $\TS(G,\omega)$ over a grid graph is in fact superlinear.

The paper is organized as follows: We give the preliminary definitions in the next section and then we 
move to other sections: one on the lower bounds by the antichain method and the other on lower bounds by the covering method. The antichain method is presented on the formulas $\SPHP$, $\PHP^m_n$, $\TS(K_n,\omega)$ and $\LOP_n$. The covering method is presented for the formulas $\TS(G,\omega)$ where $G$ is a grid graph.

\section{Preliminaries}
We use $[n]$ for the set  $\{1,2,\ldots, n\}$, $\mathbb{Z}/2$ for 
$\mathbb Z \cup (\mathbb{Z} + \frac{1}{2})$ and $\mathbb Z^+$ for $\{1,2,\ldots\}$.

\subsection{Proof systems}
Here we recall the definition of the Stabbing Planes proof system from \cite{DBLP:conf/innovations/BeameFIKPPR18}.

\begin{definition}
	A \emph{linear integer inequality} in the variables $x_1, \ldots, x_n$ is an expression of the form $\sum_{i = 1}^n a_i x_i \geq b$, where each $a_i$ and $b$ are integral. A set of such inequalities is said to be {\em unsatisfiable} if there are no $0/1$ assignments to the $x$ variables satisfying each inequality simultaneously.
\end{definition}
\noindent Note that we reserve the term infeasible, in contrast to unsatisfiable, for (real or rational) linear programs.
\begin{definition} \label{def:SP}
Fix some variables $x_1, \ldots, x_n$. A \emph{Stabbing Planes ($\SP$)} 
proof of  a set of integer linear inequalities $\mathcal{F}$ is a binary tree $\T$, with each node labeled with a \emph{query} $(\mathbf a, b)$ with $\mathbf a \in \mathbb{Z}^n, b \in \mathbb{Z}$. Out of each node we have an edge labeled with $\mathbf a x \geq b$ and the other labeled with its integer negation $\mathbf a  x \leq b-1$. Each leaf $\ell$ is labeled with a $\LP$ system $P_\ell$ made by  a nonnegative linear combination of inequalities from $\mathcal{F}$ and the inequalities labelling the edges on the path from the root of $\T$ to the leaf $\ell$. 

If $\mathcal{F}$ is an {\em unsatisfiable} set of integer linear inequalities, $\T$ is  a \emph{Stabbing Planes ($\SP$)}  {\em refutation} of $\mathcal{F}$ if all the $\LP$'s $P_\ell$ on the leaves of $\T$ are infeasible. 
 \end{definition}

\begin{definition}
	The \emph{slab}  corresponding to a query $Q=(\mathbf a, b)$ is the set $\slab(Q)=\{\mathbf x \in \mathbb{R}^n : b -1 < \mathbf a \mathbf x < b \}$ satisfying neither of the associated inequalities.
\end{definition}

Since each leaf in a $\SP$ refutation is labelled by an infeasible $\LP$, 
throughout this paper we will actually use the following geometric observation on $\SP$ proofs $\T$: the set of points in $\mathbb{R}^n$ must all be ruled out by a query somewhere in $\T$. 
In particular this will be true for those points in  $\mathbb{R}^n$  which  satisfy a set  of integer linear inequalities $\mathcal{F}$ and which we call {\em feasible points} for $\mathcal{F}$.

\begin{fact}
\label{fact:slab}
	The slabs associated with a $\SP$ refutation must cover the feasible points of $\mathcal{F}$. That is,
	
	\[
	\{\mathbf y \in \mathbb{R}^n : \mathbf a \mathbf  y \geq b \text{ for all } (\mathbf a, b) \in \mathcal{F} \} \subseteq \bigcup_{(\mathbf a, b) \in \mathcal{F}} \{\mathbf x \in \mathbb{R}^n : b -1 < \mathbf a \mathbf x < b \}
	\]
\end{fact}

The {\em length} of a $\SP$ refutation is the number of queries in the proof tree.
The {\em depth} of a $\SP$ refutation $\T$ is the longest root-to-leaf path
in $\T$.  The size (respectively depth) of refuting $\mathcal{F}$ in $\SP$ is the {\em minimum} size (respectively depth) over all $\SP$ refutations of $\mathcal{F}$.
We call {\em bit-size} of a $\SP$ refutation $\T$ the total number of bits needed to represent every inequality in the refutation. 

\begin{defiC}[\cite{CookCoullardTuran}]
	The \emph{Cutting Planes (CP)} proof system is equipped with boolean axioms and two inference rules:
$$
\begin{array}{c|c|c}
\mbox{Boolean Axioms}  & \mbox{Linear Combination} & \mbox{Rounding} \\
 \frac{ \quad \quad}{x \geq 0} \quad \frac{ \quad \quad}{-x \geq - 1} & 
 \frac{\mathbf a \mathbf x \geq c \quad \quad  \mathbf b \mathbf x \geq d}{\alpha\mathbf a \mathbf x  + \beta\mathbf b \mathbf x \geq \alpha c + \beta d }	
 & 
 \frac{\alpha \mathbf a\mathbf x \geq b}{\mathbf a \mathbf x \geq \lceil b/\alpha \rceil} 
\end{array}
$$	

\noindent where $\alpha,\beta,b \in \mathbb Z^+$ and $\mathbf a,\mathbf b \in \mathbb Z^n$. A CP refutation of some unsatisfiable set of integer linear inequalities  is a derivation of $0 \geq 1$ by the  aforementioned inference rules from the inequalities in $\mathcal F$.
\end{defiC}

 A $\CP$ refutation is {\em treelike} if the directed acyclic graph underlying the proof is a tree. The {\em length} of a $\CP$ refutation is the number of inequalities in the sequence. The {\em depth} is the length of the longest path from the root to a leaf (sink) in the graph. The {\em rank} of a $\CP$ proof is the maximal number of rounding rules used in a path of the proof graph.  The {\em size} of a $\CP$ refutation is the bit-size to represent all the inequalities in the proof. 

\subsection{Restrictions} 

Let $V=\{x_1,\ldots,x_n\}$ be a set of $n$ variables  and let  $\mathbf a \mathbf x \leq  b$ be a linear integer inequality.  
We say that a variable $x_i$ {\em appears in}, or is {\em mentioned by} a query $Q=(\mathbf a, b)$ if $a_i \not =0$ and {\em does not appear} otherwise.

A {\em restriction} $\rho$ is a function  $\rho:D \rightarrow \{0,1\}$, $D\subseteq V$.  A restriction acts on  a half-plane $\mathbf a \mathbf x\leq b$ setting 
the $x_i$'s according to $\rho$. 
Notice that the variables $x_i \in D$  do not  appear in the restricted half-plane.

By $\T\!\!\restriction_\rho$ we mean to apply the restriction $\rho$ to all the queries in a $\SP$ proof $\T$. The tree  $\T\!\!\restriction_\rho$  defines a new $\SP$ proof: if some $Q\!\!\restriction_\rho$ reduces to  $0 \leq - b$, for some $b\geq1$, then  that node becomes a leaf in  $\T\!\!\restriction_\rho$. Otherwise in $\T\!\!\restriction_\rho$ we simply branch on $Q\!\!\restriction_\rho$.  
Of course the solution space defined by the linear inequalities labelling a path in $\T\!\!\restriction_\rho$ is a subset of the solution space defined by the corresponding path in $\T$. Hence the leaves of $\T\!\!\restriction_\rho$ define an infeasible $\LP$.

We work with linear integer inequalities which are a translation of families of  CNFs $\F$. Hence when we write $\F\!\!\restriction_\rho$ we mean the applications of the restriction $\rho$ to the  set of linear integer inequalities defining $\F$.

 

   \section{The antichain method}
This method  is based on Sperner's theorem. Using it we can  prove  depth lower bounds in $\SP$ for  $\PHP^m_n$  and for Tseitin contradictions $\TS(K_n,\omega)$ over the complete graph.  To motivate and explain the main definitions,  we use as an  example a simplification of the $\PHP^m_n$, the {\em Simplified Pigeonhole principle} $\SPHP_n$, which has some interest since (as we will show) it exponentially separates $\CP$ rank from $\SP$ depth.

\subsection{Simplified Pigeonhole Principle}
As mentioned in the Introduction, the $\SPHP_n$ intends to capture 
the core of the counting argument used to efficiently refute the 
$\PHP$ in $\CP$.  
\begin{definition}
	The $\SPHP_n$ is the following unsatisfiable family of inequalities:
$$
\begin{array}{l}
\sum_{i = 1}^{n} x_i \geq 2 \\
x_i + x_j \leq 1  \qquad \text{ (for all $i \neq j \in [n]$)}
\end{array}
$$	
	
\end{definition}

\begin{lem}
	$\SPHP_n$ has a rank $1$ $\CP$ refutation, for $n \geq 3$.
\end{lem}

\begin{proof}
	Let $S := \sum_{i = 1}^{n} x_i$ (so we have $S \geq 2$). We fix some $i \in [n]$ and sum $x_i + x_j \leq 1$ over all $j \in [n] \setminus \{i\}$ to find $S + (n - 2) x_i \leq n - 1$. 
	We add this to $- S \leq - 2$ to get
	\[
	x_i \leq \frac{n - 3}{n - 2} 
	\]
	
\noindent	which becomes $x_i \leq 0$ after a single cut. We do this for every $i$ and find $S \leq 0$ - a contradiction when combined with the axiom $S \geq 2$.
\end{proof}

It is easy to see that $\SPHP_n$ has depth $O(\log n)$ proofs in $\SP$, either by a direct proof or appealing to the polynomial size proofs in $\CP$ of the $\PHP^m_n$ (\cite{CookCoullardTuran}) and then using the Theorem 4.4  in \cite{DBLP:conf/innovations/BeameFIKPPR18} informally stating that ``$\CP$ proofs can be balanced in $\SP$''.
\begin{cor}
	The $\SPHP_n$ has $\SP$ refutations of depth $O(\log n)$.
\end{cor}

We will prove that this bound is tight.

\subsection{Sperner's Theorem}

Let $\mathbf a \in \mathbb R^n$. The {\em width} $w(\mathbf a)$ of $\mathbf a$ is the number of non-zero coordinates in $\mathbf a$.  The width of a query  $(\mathbf a,b)$ is $w(\mathbf a)$, and the width of a $\SP$ refutation is the minimum width of its queries. 

Let $n \in \mathbb N$. Fix $W\subseteq [0,1] \cap \mathbb Q^+$ of finite size $k\geq 2$ and insist that $0 \in W$. The $W$'s we work with in this paper are $\{0, 1/2\}$ and $\{0, 1/2, 1\}$.

\begin{definition} 
A $(n,W)$-word is an element  in $W^n$.
\end{definition}

We consider the following  extension of Sperner's theorem.
\begin{thmC}[\cite{anti1,anti2}]
\label{thm:sperner}
	Fix any $t\geq 2, t \in \mathbb{N}$. For all $f \in \mathbb{N}$, with the pointwise ordering of $[t]^f$, any antichain has size at most $t^f \sqrt{\frac{6}{\pi (t^2 - 1) f}} (1 + o(1)).$
\end{thmC}

We will use the simplified bound that any antichain $\A$ has size $|\A| \leq \frac{t^f}{\sqrt{f}}$.

\begin{lem}
\label{lem:rulingout}
Let $\mathbf a \in \mathbb Z^n$ and $|W|=k \geq 2$. The number of $(n,W)$-words
$\mathbf s$ such that $\mathbf a \mathbf s =b$, where $b \in \mathbb Q$, is at most $\frac{k^n}{\sqrt{w(\mathbf a)}}$. 
\end{lem}

\begin{proof}
Define $I_{\mathbf a}=\{i\in [n] : a_i \not = 0\}$. Let $\preceq$ be the partial order over $W^{I_{\mathbf a}}$ where $\mathbf x \preceq  \mathbf y$ if $x_i \leq y_i$ for all $i$ with $a_i > 0$ and $x_i \geq y_i$ for the remaining $i$ with $a_i < 0$.  Clearly the set of solutions to $\mathbf a \mathbf s =b$ forms an antichain under $\preceq$. Noting that $\preceq$ is isomorphic to the typical pointwise ordering on $W^{I_{\mathbf a}}$, we appeal to \Cref{thm:sperner} to upper bound the number of solutions in $W^{I_{\mathbf a}}$ by $\frac{k^{w(\mathbf a)}}{\sqrt{w(\mathbf a)}}$, each of which corresponds to at most $k^{n - w(\mathbf a)}$ vectors in $W^n$.
\end{proof}
\subsection{Large admissibility} \label{admissibility}

A $(n,W)$-word  $s$ is {\em admissible}  for an unsatisfiable set of integer linear inequalities $\F$ over $n$ variables if $s$ satisfies all constraints of $\F$. 
A set of $(n,W)$-words is admissible for $\F$ if all its elements are admissible. 
\textcolor{black}{Let $\A(\F,W)$ be}  the set of all admissible $(n,W)$-words for $\F$.

The interesting sets $W$ for an unsatisfiable set of integer linear inequalities $\F$ are those such that almost all $(n,W)$-words are admissible for $\F$. 
We will apply our method on sets of integer linear inequalities which are a translation of unsatisfiable CNF's generated over a given domain. Typically  these formulas on a size $n$ domain have a number of variables which is not exactly $n$ but a function of $n$, $\nu(n)\geq n$. 
Hence  for the rest of this section we consider $\sF:=\{\F_{n}\}_{n \in \mathbb N}$ as a family of sets of unsatisfiable integer linear inequalities, where $\F_n$ has $\nu(n) \geq n$ variables. We call $\sF$ an  {\em unsatisfiable family}. \\

Consider then the following definition (recalling that we denote $k=|W|$):

\begin{definition} 
	\label{def:almostfull}
	$\sF$ is {\em almost full} if $|\A(\F_n,W)| \geq k^{\nu(n)}-o(k^{\nu(n)})$. 
\end{definition}

\medskip
Notice that, because of the $o$ notation, Definition \ref{def:almostfull} might be not necessarily true for all $n\in \mathbb N$, but only starting from some $n_{\sF}$.    
\begin{definition}\label{def:nLargeEnough}
	Given some almost full family $\sF$ (over $\nu(n)$ variables)  we let $n_{\sF}$ be the natural number with 
	\[
	\frac{k^{\nu(n)}}{|\A(\F_n, W)|} \leq 2 \text{\quad for all \quad} n \geq n_{\sF}.
	\]
\end{definition}
As an example we prove $\SPHP$ is almost full (notice that in the case of $\SPHP_n$, $\nu(n)=n$).
\begin{lem}
\label{lem:SPHPfull}
$\SPHP_n$ is almost full.
\end{lem}
\begin{proof}
Fix $W=\{0,1/2\}$ so that $k=|W|=2$. Let $U$ be the set of all $(n,W)$-words with at least four coordinates set to $1/2$. $U$ is admissible for $\SPHP_n$ since inequalities $x_i+x_j \leq 1$  are always satisfied for any value in $W$ and inequalities $x_1+\ldots +x_n\geq 2$ are satisfied by all points in $U$ which contain at least four $1/2$s. By a simple counting argument, in $U$ there are $2^n - 4n^3= 2^n-o(2^n)$ admissible $(n,W)$-words. Hence the claim.
\end{proof}

\begin{lem}
\label{lem:width}
Let $\sF=\{\F_{n}\}_{n \in \mathbb N}$ be an almost full unsatisfiable family, where $\F_n$ has $\nu(n)$ variables.
Further let $\T$ be a $\SP$ refutation of $\F$ of minimal width $\omega$. If $n \geq n_{\sF}$ then $|\T| = \Omega(\sqrt{w})$.
\end{lem}
\begin{proof}
We  estimate \textcolor{black}{the rate at which the slabs} of the queries in $\T$ rule out admissible points in $U$. Let $\ell$ be the least common multiple of the denominators in $W$. Every $(n, W)$-word $x$ falling in the slab of some query $(\mathbf{a}, b)$ satisfies one of $\ell$ equations $\mathbf{a} x = b + i/\ell, 1 \leq i < \ell$ (as $\mathbf{a}$ is integral).  Note that as $|W|$ is a constant independent of $n$, so is $\ell$.

 Since  all the queries in $\T$ have width at least $w$, according to Lemma \ref{lem:rulingout}, each query in $\T$ rules out at most $\ell \cdot \frac{k^{\nu(n)}}{\sqrt{w}}$ admissible points. By Fact \ref{fact:slab}  no point survives at the leaves, in particular the admissible points. Then it must be that 
$$
|\T|  \ell \cdot \frac{k^{\nu(n)}}{\sqrt{w}} \geq |\A(\F_{n}, W)| \text{ \quad which means \quad } |\T| \ell \cdot \frac{k^{\nu(n)}}{|\A(\F_{n}, W)|} \geq \sqrt{w}
$$

We finish by noting that, by the assumption $n \geq n_{\sF}$, and then by Definition \ref{def:nLargeEnough}, we have $2 \geq 
 \cdot \frac{k^{\nu(n)}}{|\A(\F_{n}, W)|}$, so $|\T| \geq \sqrt{w}/(2 \ell) \in \Omega(\sqrt{w})$.
\end{proof}
  

\subsection{Main theorem}
We focus on restrictions $\rho$ that after applied on an unsatisfiable family $\sF=\{\F_n\}_{n \in \mathbb N}$, reduce the set $\F$ to another set in the same family.

\begin{definition} 
\label{def:self} Let $\sF=\{\F_n\}_{n \in \mathbb N}$ be an unsatisfiable family and $c$ a positive constant. $\sF$  is {\em $c$-self-reducible} if for any set $V$ of variables, with $|V| = v < n / c$, there is a restriction $\rho$ with domain $V'\supseteq V$, such that 
$\F_n\!\!\restriction_\rho = \F_{n - c v}$ (up to renaming of variables).
\end{definition}

Let us motivate the definition with an example.

\begin{lem}
\label{lem:SPHPself}
$\SPHP_n$ is $1$-self-reducible.
\end{lem}
\begin{proof}
Whatever set of variables $x_i$, $i \in I \subset [n]$ we consider, it is sufficient to set $x_i$ to $0$ to fulfill Definition \ref{def:self}. 
\end{proof}

\begin{thm} \label{thm:ACLowerBound}
	Let $\sF:=\{\F_{n}\}_{n \in \mathbb N}$ be a unsatisfiable set of integer linear inequalities which is almost full and $c$-self-reducible. If $\F_n$ defines a feasible $\LP$ whenever $n > n_{\sF}$, then for $n$ large enough, the shortest $\SP$ proof of $\F_n$ is of length $\Omega(\sqrt[4]{n})$.
\end{thm}

\begin{proof}
	Take any $\SP$ proof  $\mathcal{T}$ refuting $\F_n$ and fix $t = \sqrt[4]{n}$.

	The proof proceeds by stages $i\geq 0$ where $\T_0=\T$. The stages will go on 
	while  the invariant property  (which at stage $0$ is true since $n > n_{\sF}$ and $c$ a positive constant)
	$$ n - i c t^3 > \max\{ n_{\sF}, n(1-1/c)\}$$  
	holds.
	
	At the stage $i$ we let $\Sigma_i = \{(\mathbf a,b) \in \T_i : w(\mathbf a) \leq t^2 \}$ and $s_i= |\Sigma_i|$. 
	If $s_i\geq t$ the claim is trivially proven. If $s_i=0$, then all queries in $\T_i$ have width at least $t^2$  
	and by Lemma \ref{lem:width} (which can be applied since  $n - i c t^3 >  n_{\sF}$)
	the claim is proven (for $n$ large enough).
	
	So assume that $0< s_i < t$. Each of the queries in $\Sigma_i$ involves at most $t^2$ nonzero coefficients, hence in total they mention at most $s_i t^2 \leq t^3$ variables.  Extend this set of variables to some $V'$ in accordance with \Cref{def:self} (which can be done since, by the invariant,  $i c t^3 < n/c$).
	Set all these variables according to self-reducibility of $\F$ in a restriction $\rho_i$ and define $\T_{i+1}= \T_{i}\!\!\restriction\!\!_{\rho_i}$. Note that by Definition \ref{def:self} and by that of restriction, $\T_{i+1}$ is a $\SP$ refutation of $\F_{n - ic t^3}$ and we can go on with the next stage. (Also note that we do not hit an empty refutation this way, due to the assumption that $\F_n$ defines a feasible LP.)
	
	Assume that the invariant does not hold.  If this is because $n - i c t^3 < n_{\sF}$ then, as each iteration destroys at least one node, 
	\[
	|\T| \geq i > \frac{n - n_{\sF}}{c t^3} \in \Omega(n^{1/4}).
	\]
If this is because  $n - i c t^3 <  n-n/c$, then again for the same reason it holds that 
	\begin{align*} 
	|\T| \geq i > \frac{n}{c^2n^{3/4}} \in \Omega(n^{1/4}). \tag*{\qedhere}
	\end{align*}
\end{proof}

Using Lemmas \ref{lem:SPHPfull} and \ref{lem:SPHPself} and the previous Theorem we get: 
\begin{cor}
The length of any $\SP$ refutation of $\SPHP_n$ is $\Omega(\sqrt[4]{n})$. Hence the minimal depth is $\Omega(\log n)$.
\end{cor}

\subsection{Lower bounds for the  Pigeonhole principle}

\begin{definition}
	The \emph{Pigeonhole Principle} $\PHP^m_n$, $m(n) > n$, is the family of unsatisfiable integer linear inequalities defined over the variables $\{P_{i, j} : i \in [m], j \in [n]\}$ consisting of the following inequalities:
$$
\begin{array}{ll}		
\sum_{j = 1}^n P_{i, j} \geq 1 \quad \forall i \in [m] &  \text{(every pigeon goes into some hole)} \\
		P_{i, k} + P_{j, k} \leq 1 \quad  \forall k \in [n], i \neq j \in [m] & \text{(at most one pigeon enters any given hole)}
		\end{array}
$$	
\end{definition}

We present a lower bound for $\PHP^m_n$ closely following that for $\SPHP_n$, in which we largely ignore the diversity of different pigeons (which makes the principle rather like $\SPHP_n$).

In this subsection we fix $W = \{0, 1/2\}$, and for the sake of brevity refer to $(n, W)$-words as \emph{biwords}.

In this section we fix $m$ to be $n + d$, for any fixed $d \in \mathbb{N}$ at least one.

\begin{lem}
	The $\PHP^{n+d}_n$ is almost full.
	\label{lem:PHP-fullness}
\end{lem}
\begin{proof}
	We show that there are at least $2^{mn-1}$ admissible biwords (for sufficiently large $n$).
	For each pigeon $i$, there are admissible valuations to holes so that, so long as at least two of these are set to $1/2$, the others may be set to anything in $\{0,1/2\}$. This gives at least $2^n-(n+1)$ possibilities. Since the pigeons are independent, we obtain at least $(2^n-(n+1))^m$ biwords. Now this is $2^{mn}\left(1-\frac{n+1}{2^n}\right)^m$ where $\left(1-\frac{n+1}{2^n}\right)^m \sim e^{\frac{-(n+1)m}{2^n}}$ whence, $\left(1-\frac{n+1}{2^n}\right)^m \geq e^{\frac{-(n+2)m}{2^n}}$ for sufficiently large $n$. It follows there is a constant $c$ so that:
	\[ 2^{mn}\left(1-\frac{n+1}{2^n}\right)^m \geq 2^{mn-\frac{c(n+2)m}{2^n}} \geq 2^{mn-1} \]
	for sufficiently large $n$.
\end{proof}

\begin{lem}
	The $\PHP^{n+d}_n$ is $1$-self-reducible.
	\label{lem:PHP-sr}
\end{lem}
\begin{proof}
	We are given some set $I$ of variables from \textcolor{black}{$\PHP^{n+d}_n$}. These variables will mention some set of holes $H := \{ j : P_{i, j} \in I \text{ for some i} \}$ and similarly a set of pigeons $P$. Each of $P$, $H$ have size at most $|I|$ and we extend them both arbitrarily to have size exactly $|I|$. Our restriction matches $P$ and $H$ in any way and then sets any other variable mentioning a pigeon in $P$ or a hole in $H$ to $0$.
\end{proof}

\begin{thm}
	The length of any $\SP$ refutation of $\PHP^{n + d}_n$ is $\Omega(n^{1/4})$.
	\label{thm:PHP}
\end{thm}
\begin{proof}
	Note that the all $1/2$ point is feasible for $\PHP^{n+d}_n$. Then with \Cref{lem:PHP-fullness} and \Cref{lem:PHP-sr} in hand we meet all the prerequisites for \Cref{thm:ACLowerBound}.	
\end{proof}

By simply noting that a $\SP$ refutation is a binary tree, we get the following corollary.

\begin{cor}
	The $\SP$ depth of the $\PHP^{n + d}_n$ is $\Omega(\log n)$.
\end{cor}

\subsection{Lower bounds for Tseitin contradictions over the complete graph}

\begin{definition}
	For a graph $G = (V, E)$ along with a charging function $\omega: V \to \{0, 1\}$  satisfying $\sum_{v \in V} \omega(v) = 1 \mod 2$. The \emph{Tseitin contradiction} $\TS(G,\omega)$ is the set of linear inequalities which translate the CNF encoding of 
	\[
	\sum_{\substack{ e \in E\\ e \ni v}} x_e = \omega(v) \mod 2.
	\]
	for every $v \in V$, where the variables $x_e$ range over the edges $e \in E$.
	
\end{definition}

In this subsection we consider $\TS(K_n,\omega)$ and $\omega$ will always be an odd charging for $K_n$. 
We let $N := \binom{n}{2}$ and we fix $W = \{0, 1/2, 1\}$, $k=3$ and for the sake of brevity refer to $(n, W)$-words as \emph{triwords}. We will abuse slightly the notation of \Cref{admissibility} and consider the family $\{\TS(K_n, \omega) \}_{n \in \mathbb{N}, \; \omega \text{ odd}}$ as a single parameter family in $n$. The reason we can do this is because the following proofs of almost fullness and self reducibility do not depend on $\omega$ at all (so long as it is odd, which we will always ensure).

\begin{lem}
	$\TS(K_n, \omega)$ is almost full.
\end{lem}
\begin{proof}
	
	We show that $\TS(K_n, \omega)$ has at least $c 3^N$ admissible triwords, for any constant $0 < c < 1$ and $n$ large enough.
	We define the assignment $\rho$ setting all edges (i.e. $x_e$) 
	 to a value in $W=\{0,1,1/2\}$ independently and uniformly at random, and inspecting the probability that  
	  some fixed constraint for a node $v$ is violated  by $\rho$.
	
	Clearly if at least $2$ edges incident to $v$ are set to $1/2$ its constraint is satisfied. If none of its incident edges are set to $1/2$ then it is satisfied with probability $1/2$. Let $A(v)$ be the event ``{\em no edge incident to $v$ is set to $1/2$ by $\rho$}'' and let $B(v)$ be the event that ``{\em exactly one edge incident to $v$ is set to $1/2$ by $\rho$}''.
	Then: 
	\begin{gather*}
		\Pr[\text{$v$ is violated}] \leq \frac{1}{2} \Pr[A(v)] + \Pr[B(v)] 
		= \frac{1}{2} \frac{2^{n - 1}}{3^{n - 1}} + \frac{(n-1)2^{n - 2}}{3^{n - 1}} = n \frac{2^{n - 2}}{3^{n - 1}}.
	\end{gather*}

Therefore, by a union bound, the probability that there exists a node with violated parity is bounded above by $n^2 \frac{2^{n - 2}}{3^{n - 1}}$, which approaches $0$ as
	$n$ goes to infinity.
\end{proof}

\begin{lem}
	$\TS(K_n, \omega)$ is $2$-self-reducible.
\end{lem}

\begin{proof}
	We are given some set of variables $I$. Each variable mentions $2$ nodes, so extend these mentioned nodes arbitrarily to a set $S$ of size exactly $2 |I|$, which we then
	hit with the following restriction: if $S$ is evenly charged, pick any matching on the set $\{s \in S : w(s) = 1\}$, set those edges to $1$, and set any other edges involving some vertex in $S$ to $0$. Otherwise (if $S$ is oddly charged) pick any $l \in \{s \in S : w(s) = 1\}$ and $r \in [n] \setminus S$ and set $x_{l r}$ to $1$. $\{s \in S : w(s) = 1\} \setminus l$ is now even so we can pick a matching as before. And as before we set all other edges involving some vertex in $S$ to 0. In the first case the graph induced by $[n] \setminus S$ must be oddly charged (as the original graph was). In the second case this induced graph was originally evenly charged, but we changed this when we set $x_{l r}$ to $1$.
\end{proof}

\begin{lem}
	For any oddly charged $\omega$ and $n$ large enough, all $\SP$ refutations of $\TS(K_n, \omega)$ have length $\Omega(\sqrt[4]{n})$.
\end{lem}
\begin{proof}
	We have that the all $1/2$ point is feasible for $\TS(K_n, \omega)$. Then we can simply apply \Cref{thm:ACLowerBound}.
\end{proof}

\begin{cor}
	The depth of any SP refutation of $\TS(K_n, \omega)$ is $\Omega(\log{n})$.
\end{cor}

	\subsection{Lower bound for the Least Ordering Principle}
	
	\begin{definition}
		Let $n \in \mathbb{N}$. The \emph{Least Ordering Principle}, $\LOP_n$, is the following set of unsatisfiable linear inequalities over the variables $P_{i, j}$ ($i \neq j \in [n]$):
		\begin{gather*}
			P_{i, j} + P_{j, i} = 1 \quad \text{ for all $i \neq j \in [n]$} \\
			P_{i, k} - P_{i, j} - P_{j, k} \geq \textcolor{black}{-1} \text { for all $i \neq j \neq k \in [n]$} \\
			\sum_{i = 1, i \neq j}^n P_{i, j} \geq 1 \text{ for all $j \in [n]$}
		\end{gather*}
	\end{definition}

	\begin{lem} \label{lem:LOPMentioned}
		For any $X \subseteq [n]$ of size at most $n-3$, there is an admissible point for $\LOP_n$ integer on any edge mentioning an element in $X$.
	\end{lem}
	\begin{proof}
		Let $\preceq$ be any total order on the elements in $X$. Our admissible point $x$ will be
		\[
		x(P_{i, j}) = \begin{cases}
			1 & \text{ if $i, j \in X$ and $i \preceq j$, or if $i \not\in X, j \in X$} \\
			0 & \text{ if $i, j \in X$ and $j \preceq i$, or if $i \in X, j \not\in X$} \\
			1/2 & \text{otherwise (if $i, j \not\in X$)}.
		\end{cases}
		\]
		
		The existential axioms $\sum_{i = 1, i \neq j}^n P_{i, j}$ are always satisfied - if $j \in X$ then there is some $i \not\in X$ with $P_{i, j} = 1$, and otherwise there are at least two distinct $i, k \neq j \in X$ with $P_{i, j}, P_{k, j} = 1/2$. For the transitivity axioms $P_{i, k} - P_{i, j} - P_{j, k} \geq \textcolor{black}{-1}$, note that if $2$ or more of $i, j, k$ are not in $X$ there are at least $2$ variables set to $1/2$, and otherwise it is set in a binary fashion to something consistent with a total order.
\end{proof}
	
	We will assume that a $\SP$ refutation $\T$ of $\LOP_n$ only involves variables $P_{i,j}$ where $i<j$ - this is without loss of generality as we can safely set $P_{j,i}$ to $1 - P_{i,j}$ whenever $i>j$, and will often write $P_{\{i,j\}}$ for such a variable. We consider the underlying graph of the support of a query, i.e. an undirected graph with edges $\{i,j\}$ for every variable $P_{\{i,j\}}$ that appears with non-zero coefficient in the query.
	
	For some function $f(n)$, we say the query is \emph{$f(n)$-wide} if the smallest edge cover of its graph has at least $f(n)$ nodes . A query that is not $f(n)$-wide is \emph{$f(n)$-narrow}. The next lemma works much the same as \Cref{thm:ACLowerBound}.

	\begin{lem} \label{lem:LOPWidth}
		Fix $\epsilon > 0$ and suppose we have some $\SP$ refutation $\T$ of $\LOP_n$, where $|\T| \leq n^{\frac{1 - \epsilon}{4}} $. Then, if $n$ is large enough, we can find some $\SP$ refutation $\T'$ of $\LOP_{c \cdot n}$, where $c$ is a positive universal constant that may be taken arbitrarily close to $1$, $\T'$ contains only $n^{3/4}$-wide queries, and $|\textcolor{black}{\T'}| \leq |\T|$.
	\end{lem}
	\begin{proof}
		We iteratively build up an initially empty restriction $\rho$. At every stage $\rho$ imposes a total order on some subset $X \subseteq [n]$ and places the elements in $X$ above the elements not in $X$. So $\rho$ sets every edge not contained entirely in $[n] \setminus X$ to something binary, and $\LOP_n\!\!\restriction_\rho = \LOP_{n - |X|}$ (up to a renaming of variables).
		
		While there exists a $n^{3/4}$-narrow query $q \in \T\!\!\restriction_\rho$ we simply take its smallest edge cover, which has size at most $n^{3/4}$ by definition, and add its nodes in any fashion to the total order in $\rho$. Now all of the variables mentioned by $q \in \T\!\!\restriction_\rho$ are fully evaluated and $q$ is redundant. We repeat this at most  $n^\frac{1 - \epsilon}{4}$ times (as $|\T| \leq n^\frac{1 - \epsilon}{4}$ and each iteration renders at least one query in $\T$ redundant). At each stage we grow the domain of the restriction by at most $n^{3/4}$, so the domain of $\rho$ is always bounded by $n^{1 - \epsilon/4}$. We also cannot exhaust the tree $\T$ in this way, as otherwise $\T$ mentioned at most $n^{1 - \epsilon/4} < n - 3$ elements and by \Cref{lem:LOPMentioned} there is an admissible point not falling in any slab of $\T$, violating \Cref{fact:slab}. \\
		
		When this process finishes we are left with a $n^{3/4}$-wide refutation $\T'$ of $\LOP_{n - n^{1 - \epsilon/4}}$. As $\epsilon$ was fixed we find that as $n$ goes to infinity $n - n^{1 - \epsilon/4}$ tends to $n$.
	\end{proof}

	\begin{lem}\label{lem:LOPHypercube}
		Let $d \leq (n-3)/2$. Given any disjoint set of pairs $D = \{ \{l_1, r_1\}, \ldots,$ $\{ l_d, r_d\} \}$ (where without loss of generality $l_i < r_i$ in $[n]$ as natural numbers) and any binary assignment $b \in \{0,1\}^D$, the assignment $x_b$ with 
		
		\[
		x_b(P_{\{i, j\}}) = \begin{cases}
			b(\{l_k, r_k\}) & \text{ if $\{i, j\} = \{l_k, r_k\} \in \textcolor{black}{D}$ for some $k$} \\
			1/2 & \text{otherwise}
		\end{cases}
		\]
		is admissible.
	\end{lem}
	\begin{proof}
	 	The existential axioms $\sum_{i = 1, i \neq j}^n P_{i, j}$ are always satisfied, as for any $j$ there are at least $n-2$  $i \in [n]$ different from $j$ with $P_{i, j} = 1/2$. For the transitivity axioms $P_{i, k} - P_{i, j} - P_{j, k} \geq \textcolor{black}{-1}$, note that due to the disjointness of $D$ at least two variables on the left hand side are set to $1/2$.
	\end{proof}

	\begin{thm}
		Fix some $\epsilon > 0$ and let $\T$ any $\SP$ refutation of $\LOP_n$. Then, for $n$ large enough, $|\T| \in \Omega(n^{\frac{1 - \epsilon}{4}})$.
	\end{thm}
	\begin{proof}
		Suppose otherwise - then, by \Cref{lem:LOPWidth}, we can find some $\T'$ refuting $\LOP_{c n}$, with $|\T'| \leq |\T|$, every query $n^{3/4}$-wide, and $c$ independent of $n$. We greedily create a set of pairs $D$ by processing the queries in $\T'$ one by one and choosing in each a matching of size $n^{1/2}$ disjoint from the elements appearing in $D$ - this always succeeds, as at every stage $|D| \in O(n^\frac{1 - \epsilon}{4} \cdot n^{1/2})$ and involves at most $O(2 n^\frac{3 - \epsilon}{4}) < n^{3/4} - n^{1/2}$ elements.\\

		So by \Cref{lem:LOPHypercube}, after setting every edge not in $D$ to $1/2$, we have some set of linear polynomials $\R = \{a(x) = \mathbf{a} x - b - 1/2 : (\overline{a}, b) \in \T'\}$ covering the hypercube $\{0, 1\}^D$, where every polynomial $p \in \R$ mentions at least $n^{1/2}$ edges. By \Cref{lem:rulingout} each such polynomial in $\R$ rules out at most $\nicefrac{2^{|D|}}{n^{1/4}}$ points, and so we must have $|\T| \geq |T'| \geq |\R| \geq n^{1/4}$.
	\end{proof}

\section{The covering method}
\begin{definition}
		A set $L$ of linear polynomials with real coefficients is said to be a \emph{cover} of the cube $\{0,1\}^n$ if
		for each $v \in \{0,1\}^n$, there is a $p \in L$ such that $p(v)=0$.
\end{definition}

In \cite{essential} Linial and Radhakrishnan  considered the problem  of the minimal number of hyperplanes needed to cover the cube $\{0,1\}^n$. Clearly every such cube can be covered by the zero polynomial, so to make the problem more meaningful they defined the notion 
of an {\em essential covering} of $\{0,1\}^n$.

\begin{defiC}[\cite{essential}] \label{def:essential}
	A set $L$ of linear polynomials with real coefficients is said to be an \emph{essential cover} of the cube $\{0,1\}^n$ if
	\begin{enumerate}[align=left]
		\item[(E1)]$L$ is a cover of $\{0,1\}^n$, 
		\item[(E2)] no proper subset of $L$ satisfies (E1), that is, for every $p \in L$, there is a $v \in \{0,1\}^n$ such that $p$ alone takes the value $0$ on $v$, and
		\item[(E3)] every variable appears (in some monomial with non-zero coefficient) in some polynomial of $L$.
	\end{enumerate}
\end{defiC}

They then proved that any essential cover $E$ of the hypercube $\{0,1\}^n$ must satisfy $|E| \geq \sqrt{n}$. We will use the slightly strengthened lower bound given in \cite{essentialLB}:

\begin{thm} \label{thm:essentialLB}
	Any essential cover $L$ of the cube with $n$ coordinates satisfies $|L| \in \Omega(n^{0.52})$.
\end{thm}

We will need an auxillary definition and lemma.

\begin{definition}
	Let $L$ be a cover of $\{0, 1\}^{I}$ for some index set $I$. Some subset $L'$ of $L$ is an \emph{essentialisation} of $L$ if $L'$ also covers $\{0, 1\}^{I}$ but no proper subset of it does.
\end{definition}
\begin{lem}
	Let $L$ be a cover of the cube $\{0,1\}^n$  and $L'$ be any essentialisation of $L$. Let $M'$ be the set of variables appearing with nonzero coefficient in $L'$. Then $L'$ is an essential cover of $\{0, 1\}^{M'}$.
\end{lem}
\begin{proof} \textcolor{black}{Let $M=[n]$.}
	\begin{itemize}[align=left]
		\item[(E1)] Given any point $x \in \{0, 1\}^{M'}$, we can extend it arbitrarily to a point $x' \in \{0, 1\}^{M}$. Then there is some $p \in L'$ with $p(x') = 0$ - but $p(x') = p(x)$, as $p$ doesn't mention any variable outside of $M'$.
		
		\item[(E2)] Similarly to the previous point, this will follow from the fact that if some set $\T$ covers a hypercube $\{0, 1\}^I$, it also covers $\{0, 1\}^{I'}$ for any $I' \supseteq I$.
		
		Suppose some proper subset $L'' \subset L'$ covers $\{0, 1\}^{M'}$, then it covers $\{0, 1\}^{M}$ - but we picked $L'$ to be a minimal set with this property.
		
		\item[(E3)] We defined $M'$ to be the set of variables appearing with nonzero coefficient in $L'$. \qedhere
	\end{itemize}
\end{proof}

\subsection{The covering method and Tseitin}

Let $H_n$ denote the $n \times n$ grid graph. Fix some $\omega$ with odd charge and a $\SP$ refutation $\T$ of $\TS(H_n, \omega)$. \Cref{fact:slab} tells us that for every point $x$ admissible for $\TS(H_n, \omega)$, there exists a query $(\mathbf{a}, b) \in \T$ such that $b < \mathbf{a} x < b + 1$. In this section we will only consider admissible points with entries in $\{0, 1/2, 1\}$, turning the slab of a query $(\mathbf{a}, b)$ into the solution set of the single linear equation 
$\mathbf{a} \cdot x = b + 1/2$. So we consider $\T$ as a set of such equations.

We say that an edge of $H_n$ is {\em mentioned}
in $\T$ if the variable $x_e$ appears with non-zero coefficient in some 
query in $\T$. We can see $H_n$ as a set of $(n - 1)^2$ squares (4-cycles), and we can index them as if they were a Cartesian grid, starting from $1$. Let $S$ be the set of $\lfloor (n/3)^2 \rfloor$ squares in $H_n$ \textcolor{black}{obtained} by picking squares with indices that become $2 \text{ (mod } 3)$. This ensures that every two squares in $S$ in the same row or column have at least two other squares between them, and that no selected square is on the perimeter.

We will assume without loss of generality that $n$ is a multiple of 3, so $|S| = (n/3)^2$. Let $K = \bigcup_{t \in S} t$ be the set of edges mentioned by $S$, and for some $s \in S$, let $K_s :=\bigcup_{t \in S, t \neq s} t$ be the set of edges mentioned in $S$ by squares other than $s$.

\begin{lem} \label{lem:bys}
	For every $s \in S$ we can find an admissible point $b_s \in \{0, 1/2, 1\}^{E(H_n)}$ such that
	\begin{enumerate}
		\item $b_s(x_e) = 0$ for all $e \in K_s$, and
		\item $b_s$ is fractional only on the edges in $s$.
	\end{enumerate}
\end{lem}
\begin{proof}
	We use the following fact due to A. Urquhart in \cite{DBLP:journals/jacm/Urquhart87}
	\begin{fact}
		For each vertex $v$ in $H_n$ there is a totally binary assignment, called $v$-critical in \cite{DBLP:journals/jacm/Urquhart87}, satisfying all parity axioms  in  $\TS(H_n, \omega)$ except the parity axiom of node $v$.
	\end{fact} 
	
	Pick any corner $c$ of $s$. Let $b_s$ be the result of taking any $c$-critical assignment of the variables of $\TS(H_n, \omega)$ and setting the edges in $s$ to $1/2$. $b_s$ is admissible, as $c$ is now adjacent to two variables set to $1/2$ (so its originally falsified parity axiom becomes satisfied) and every other vertex is either unaffected or also adjacent to two $1/2$s. While $b_s$ sets some edge $e \in K_s$ to 1, flip all of the edges in the unique other square containing $e$. This other square always exists (as no square touches the perimeter) and also contains no other edge in $K_s$ (as there are at least two squares between any two squares in $S$). Flipping the edges in a cycle preserves admissibility, as every vertex is adjacent to $0$ or $2$ flipped edges.
\end{proof}

\begin{definition} \label{def:tseitinSub}
	Let $V_S := \{v_s : s \in S\}$ be a set of new variables.  For $s \in S$ define the substitution $h_s$, taking the variables of $\TS(H_n, \omega)$ to $V_S \cup \{0, 1/2, 1\}$, as
	
	\[
	h_s(x_e) := \begin{cases}
		b_s(e) & \text{ if $e$ is not mentioned in $S$, or if $e$ is mentioned by $s$,} \\
		v_t & \text{if $e$ is mentioned by some square $t \neq s \in S$.} \\
	\end{cases}
	\]
	
	(where $b_s$ is from \Cref{lem:bys}).
\end{definition}

\begin{definition}
	Say that a linear polynomial $p  = c + \sum_{e \in E(H_n)} \mu_e x_e$ with coefficients $\mu_e \in \mathbb{Z}$ and some constant part $c \in \mathbb{R}$ has \emph{odd coefficient in $X \subseteq E(H_n)$} if $\sum_{e \in X} \mu_e$ is an odd integer. Given some polynomial $p$ in the variables $x_e$ of Tseitin, and some square $s \in S$, let $p_s$ be the polynomial in variables $V_S$ \textcolor{black}{obtained} by applying the substitution $x_e \to h_s(x_e)$. Also, for any set of polynomials $\T$ in the variables $x_e$ let $\T_s := \{p_s : p \in \T, p \text{ has odd coefficient in } s \}.$
\end{definition}

Given some assignment $\alpha \in \{0, 1\}^{V_S \setminus \{v_s\}}$, and some $h_s$ as in \Cref{def:tseitinSub}, we let $\alpha(h_s)$ be the assignment to the variables of $\TS(H_n, \omega)$ gotten by replacing the $v_t$ in the definition of $h_s$ by $\alpha(v_t)$.

\begin{lem}
	Let $s \in S$. For all $2^{|S| - 1}$ settings $\alpha$ of the variables in $V_S \setminus \{s\}$, $\alpha(h_s)$ is admissible.
\end{lem}
\begin{proof}
	When $\alpha(v_t)$ is all $0$, $h_s = b_s$ is admissible (by \Cref{lem:bys}). Toggling some $v_t$ only has the effect of flipping every edge in a cycle, which preserves admissibility.
\end{proof}

\begin{lem} \label{lem:TCovered}
\textcolor{black}{Let $\T$ be an $\SP$ refutation of $\TS(H_n, \omega)$.}	$\T_s$ covers $\{0, 1\}^{V_S \setminus \{s\}}$.
\end{lem}
\begin{proof}
	For every setting of $\alpha \in \{0, 1\}^{V_S \setminus \{s\}}$, $\alpha(h_s)$ as defined above is admissible and therefore covered by some $p \in \T$, which has constant part $1/2 + b$ for some $b \in \mathbb{Z}$. Furthermore, as $\alpha(h_s)$ sets every edge in $s$ to $1/2$, every such $p$ must have odd coefficient in front of $s$ - otherwise
	\[
	p(\alpha(h_s)) = 1/2 + b + (1/2) \left( \sum_{e \in s} \mu_e \right) + \sum_{e \not\in s} \mu_e \alpha(h_s)(x_e)
	\]
	can never be zero, as the $1/2$ is the only non integral term in the summation.
\end{proof}

\begin{thm}
	Any $\SP$ refutation $\T$ of $\TS(H_n, \omega)$ must have $|\T| \in \Omega(n^{1.04})$.
\end{thm}
\begin{proof}
	
	We are going to find a set of pairs $(L_1, M_1), (L_2, M_2), \ldots, (L_q, M_q)$, where the $L_i$ are pairwise disjoint nonempty subsets of $\T$, the $M_i$ are subsets of $V_S$, and for every $i$ there is some $s_i \in S \setminus \bigcup_{i = 1}^q M_i$ such that $|(L_i)_{s_i}| \geq |M_i|^{0.52}$. These pairs will also satisfy the property that
	\begin{equation} \label{eq:SCovered}
		\{ s_i : 1 \leq i \leq q \} \cup \bigcup_{i = 1}^q M_i = S.
	\end{equation}
	
	As $|S| = (n/3)^2$ this would imply that $\sum_{i = 1}^q |M_i| \geq (n/3)^2 - q$. If $q \geq (n/3)^2 / 2$, then (as the $L_i$ are nonempty and pairwise disjoint) we have $|\T| \geq (n/3)^2 / 2 \in \Omega(n^{1.04})$. Otherwise $\sum_{i = 1}^q |M_i| \geq (n/3)^2/2$, and as (by \Cref{thm:essentialLB}) each $|L_i| \geq |M_i|^{0.52}$,
	
	\begin{equation}\label{eq:essentialMainEq}
	|\T| \geq \sum_{i = 1}^q |L_i| \geq \sum_{i = 1}^q |M_i|^{0.52} \geq \left( \sum_{i = 1}^q |M_i| \right)^{0.52} \geq \left( (n/3)^2/2 \right)^{0.52} \in \Omega(n^{1.04}).
	\end{equation}
	
	We create the pairs by stages. Let $S_1 = S$ and start by picking any $s_1 \in S_1$. By \Cref{lem:TCovered} $\T_{s_1}$ covers $\{0, 1\}^{V_{S_1} \setminus \{s_1\}}$ and has as an essentialisation $E$, which will be an essential cover of $\{0, 1\}^{V'}$ for some $V' \subseteq V_{S_1} \setminus \{s_1\}$. We create the pair $(L_1, M_1) = (\{p : p_{s_1} \in E\}, V')$ and update $S_2 = S_1 \setminus \left( V' \cup \{s_1\} \right)$. (Note that $V'$ could possibly be empty - for example, if the polynomial $x_e = 1/2$ appears in $\T$, where $e \in s_1$. In this case however we still have $|L_1| \geq |M_1|^{0.52}$. If $V'$ is not empty we have the same bound due to \Cref{thm:essentialLB}.)  If $S_2$ is nonempty we repeat with any $s_2 \in S_2$, and so on.
	
	We now show that as promised the left hand sides of these pairs partition a subset of $\T$, which will give us the first inequality in \Cref{eq:essentialMainEq}. Every polynomial $p$ with $p_{s_i} \in  L_i$ has every $v_t$ mentioned by $p_{s_i}$ removed from $S_j$ for all $j \geq i$, so the only way $p$ could reappear in some later $L_j$ is if $p_{s_j} \in \T_{s_j}$, where $v_{s_j}$ does \emph{not} appear in $p_{s_i}$. Let $\mu_e, e \in s_j$ be the coefficients of $p$ in front of the four edges of $s_j$. The coefficient in front of $v_{s_j}$ in $p_{s_i}$ is just $\sum_{e \in s_j} \mu_e$. As $v_{s_j}$ failed to appear this sum is $0$ and $p$ does not have the odd coefficient sum it would need to appear in $\T_{s_j}$.
\end{proof}

   \section{Conclusions and acknowledgements}
    The $\Omega(\log n)$ depth lower bound for $\TS(H_n,\omega)$ is 
    not optimal since \cite{DBLP:conf/innovations/BeameFIKPPR18} proved an $O(\log^2 n)$ upper bound for  $\TS(G,\omega)$, for any bounded-degree $G$.  
    Even to apply the covering method to prove a  depth $\Omega(\log^2 n)$  lower bound on $\TS(K_n,\omega)$  (notice that it would imply a superpolynomial length lower bound), the polynomial covering of the boolean cube should be  improved to work on general cubes. To this end the algebraic method used in \cite{essential} should be improved to work with generalizations of multilinear polynomials.  

     \textcolor{black}{We use essential covering of the Boolean cube to prove size lower bounds in $\SP$. However our lower bounds are quite weak, in fact almost linear. It would be very interesting to understand whether the essential covering technique can prove stronger size lower bounds in $\SP$.  Notice that  any polytope in $[0, 1]^n$ can be covered by $n$ hyperplanes. But the polytopes produced  by the Stabbing Planes procedure are more specific and in fact they might require a weaker form of covering. For example a recursive covering, where the slabs on one branch do not affect the points on a different independent branch. Exploring this and similar ideas might eventually lead to improve our lower bounds.}
     
    While finishing the writing of this manuscript we learned about \cite{Fleming21} from Noah Fleming. We would like to thank him for answering some questions on his paper \cite{DBLP:conf/innovations/BeameFIKPPR18}, and sending us the manuscript \cite{Fleming21} and for comments on a preliminary  version of this work.  

    We are grateful also to several anonymous referees on both the conference and journal versions of this paper.


%
%


\bibliographystyle{alphaurl}
\bibliography{refs.bib}

\end{document}